\documentclass[reqno]{amsart}
\usepackage{amsfonts,amssymb,amsmath,amsthm}
\usepackage{url}
\usepackage{enumerate}
\usepackage{xy}
\usepackage[toc,page]{appendix}
\usepackage{color}
\usepackage{nameref}
\usepackage[dvipsnames]{xcolor}
\usepackage[colorlinks=true,linkcolor=RoyalBlue,citecolor=PineGreen,urlcolor=blue]{hyperref}
\usepackage{comment}
\usepackage{enumitem}
\usepackage[margin=3cm]{geometry}
\usepackage[ruled,vlined,linesnumbered]{algorithm2e}
\usepackage{tikz}
\usetikzlibrary{patterns}
\xyoption{all}
\usepackage{longtable}
\usepackage{subcaption}
\def\shift{6}

\newtheorem{theorem}{Theorem}
\newtheorem{proposition}{Proposition}
\newtheorem{lemma}{Lemma}

\theoremstyle{definition}

\newtheorem{remark}{Remark}

\DeclareMathOperator{\DP}{dp}

\title[Power series expansions for the planar monomer-dimer problem]{Power series expansions for \\ the planar monomer-dimer problem}

\author{Gleb Pogudin}
\address{Institute for Algebra, Johannes Kepler University, Linz, Austria}
\email{pogudin.gleb@gmail.com}

\thanks{This work was supported by the Austrian Science Fund FWF grant Y464-N18.}

\begin{document}

\begin{abstract}
	We compute the free energy of the planar monomer-dimer model.
    Unlike the classical planar dimer model, an exact solution is not known in this case.
    Even the computation of the low-density power series expansion requires heavy and nontrivial computations.
    Despite of the exponential computational complexity, we compute almost three times more terms than were previously known.
    Such an expansion provides both lower and upper bound for the free energy, and allows to obtain more accurate numerical values than previously possible.
    We expect that our methods can be applied to other similar problems.
\end{abstract}

\keywords{the monomer-dimer moodel; symbolic computation; cluster expansion; the free energy}

\maketitle

\section{Introduction}

The exact solution of the closed-packed dimer plane model obtained in \cite{Kasteleyn,Fisher,FisherTemperley} is a fundamental result in statistical mechanics and combinatorics.
In particular, it implies that the number of tilings of $m \times n$ rectangle using dimers grows as $e^{\frac{G}{\pi} mn}$, where $G \approx 0.916$ is the Catalan constant.
Similar results were later obtained for other shapes (see~\cite{CohnKenyonPropp} and references therein).
Applications to physics suggest two natural further questions: what if the dimension of the lattice is higher (i.e. we compute the number of tilings of a hyperrectangle), and what if we consider tilings using both dimers and monomers with a fixed proportion.
For the first question, we refer the reader to~\cite{Ciucu,GamarnikKatz,Abert2015} and references therein.
The second questions originates from the study of liquid mixtures on crystal surfaces in~\cite{FowlerRushbrooke}, see also~\cite{Everett164} for comparison with experimental data.
Monomer-dimer systems also arise in connection with the Ising model and the Heisenberg model, see~\cite[Sect.~5]{HeilmannLieb}.
For both these questions, the exact solution is out of reach so far.
However, even finding the answer numerically leads to very challenging computational problems, because the underlying combinatorial counting problems are very hard, 
and even a small change of the parameters of the problem makes computations much harder or even unfeasible.
For example, in~\cite{Jerrum} it is proved that the monomer-dimer tilings counting problem is \#P-complete in the sense of theoretical computer science.

In this paper we will focus on the second question, namely on the case of planar monomer-dimer tilings with fixed dimer density.
Let us state the problem precisely.
We denote by $a_p(m, n)$ the number of tilings of an $m \times n$ rectangle by monomers and dimers with exactly $\lfloor pmn / 2 \rfloor$ dimers,
where $m, n$ are positive integers, and $p \in [0, 1]$.
Then $p$ is roughly the fraction of the area covered by dimers.
We are interested in the limit
\[
f_2(p) = \lim\limits_{n, m \to \infty} \frac{\ln a_{p}(mn)}{mn}.
\]
In other words, we want to determine the constant $\lambda$ such that $a_p(m, n) \sim e^{\lambda mn}$ as a function of $p$.
From the point of view of statistical mechanics, $f_2(p)$ is equal to the negative of the Helmholtz free energy 
per lattice site expressed in units of the thermal energy $k_BT$.
Some lower and upper bounds for $f_2(p)$ were rigorously proved in~\cite{FriedlandPeled,FriedlandKropLundowMarkstrom}.
However, these bound are not very tight.

Another approach, taken in a series of papers~\cite{Federbush,FederbushFriedland,ButeraFederbushPernici} and independently in~\cite[IV.A]{Kong2} is to expand this function as a power series in $p$  (in some of these papers also expansion with respect to the dimension was discussed).
In the former papers, the authors look for a representation of $f_2(p)$ of the form
\begin{equation}\label{eq:lower_bound}
	f_2(p) = \frac{1}{2} \left( (2\ln 2 - 1)p - p\ln p\right) - (1 - p)\ln (1 - p) + \sum\limits_{j = 2}^\infty a_j p^j.
\end{equation}
In~\cite{Kong2}, the representation is of the form (see details in Section~\ref{sec:reduction_to_grand_canonical})
\begin{equation}\label{eq:upper_bound}
	f_2(p) = \frac{1}{2} \left( (2\ln 2 + 1)p - p\ln p\right) + \sum\limits_{j = 2}^\infty b_j p^j.
\end{equation}
Expanding $(1 - p)\ln(1 - p)$ into a Taylor series in $p$, it is easy to move back and forth between~\eqref{eq:lower_bound} and~\eqref{eq:upper_bound} (see~\eqref{eq:bk_via_ak}).
An important observation is that $a_j > 0$ and $b_j < 0$ for all known $a_j$ and $b_j$.
Under the assumption that this pattern holds for all $j$, truncations of~\eqref{eq:lower_bound} and~\eqref{eq:upper_bound} provide 
lower and upper bound for $f_2(p)$, respectively.
Thus, computing more terms of these series would result in tighter bounds for $f_2(p)$.

Previously, the record in the number of computed terms in~\eqref{eq:lower_bound} and~\eqref{eq:upper_bound} was $23$ (i.e. from $a_2$ till $a_{24}$)
obtained in~\cite[Table II]{ButeraFederbushPernici}. This result is highly nontrivial, because the underlying combinatorial problem has exponential complexity,
so the cost of every next term is usually higher by some factor.
In this paper, we compute $63$ terms (from $a_2$ to $a_{64}$, see Table~\ref{table:ak}) i.e., is almost three times more than what was previously possible.

The contribution of our paper is two-fold.
First, our approach allows us to compute significantly more terms for both series~\eqref{eq:lower_bound} and~\eqref{eq:upper_bound} than were previously known and,
combining them, we obtain very accurate values of $f_2(p)$.
Moreover, we provide additional support for the important conjecture that $a_j > 0$ and $b_j < 0$ for all $j \geqslant 2$.
Second, we show how methods of computer algebra (guessing, modular computation, etc.) can be applied to
study models in statistical mechanics. 
We expect that our methods can be used in other problems of this type (monomer-polymer mixtures, other types of lattices, etc.) in order to push computational limits further.

The rest of the paper is organized as follows.
In Section~\ref{sec:reduction_to_grand_canonical} we collect some known results and approaches that connect power series expansion
of $f_2(p)$ to the combinatorial data.
Section~\ref{sec:combinatorial} contains Theorem~\ref{th:main_combinatorial}, a main combinatorial ingredient of our computation.
Section~\ref{sec:algorithm} contains the description of our algorithm together with all computer algebra machinery used to speed it up.
Finally, in Section~\ref{sec:numerical} we describe our implementation, provide numerical results, and compare them to the previous work.

The author is grateful to Manuel Kauers and Doron Zeilberger for introducing him to the subject and for helpful discussions, and to the referees for the comments on the manuscript.


\section{Reduction to grand-canonical partition function}\label{sec:reduction_to_grand_canonical}

For every fixed $m$ and $n$, consider the \emph{grand-canonical partition function}
\begin{equation}\label{eq:grand_canonical_finite}
\Theta_{m, n}(z) = \sum\limits_{s = 0}^{\lfloor mn / 2 \rfloor} a_{2s / mn}(m, n)z^s.
\end{equation}
We consider a thermodynamic limit of $\Theta_{m, n}(z)$ (its existence is proved in~\cite[VIII]{HeilmannLieb})
\begin{equation}\label{eq:theta_def}
\Theta(z) = \lim\limits_{m, n \to \infty} \left( \Theta_{m, n}(z) \right)^{\frac{1}{mn}}, \text{ and } \ln\Theta(z) = \lim\limits_{m, n \to \infty} \frac{\ln \Theta_{m, n}(z)}{mn}.
\end{equation}
Since $\Theta_{m, n}'(0) = 2mn - m - n$, $\ln\Theta(z) = 2z + O(z^2)$.
Theorems~\cite[8.8A, 8.8B]{HeilmannLieb} rewritten in our notation 
(i.e., replacing $\mu$ with $\ln z$, $g(\mu)$ with $\ln\Theta(z)$, $\rho$ with $\frac{p}{2}$, and $h(\rho)$ with $f_2(p)$) state that
\begin{equation}\label{eq:legendre_first}
	f_2(p) = \inf\limits_{z \in \mathbb{R}_{+}} \left\{ - \frac{p}{2} \ln z + \ln\Theta(z) \right\}.
\end{equation}
We compute $z$, where the expression $- \frac{p}{2} \ln z + \ln\Theta(z)$ is minimal
\[
\frac{\operatorname{d}}{\operatorname{dz}}\left( - \frac{p}{2} \ln z + \ln\Theta(z) \right) = -\frac{p}{2z} + \left( \ln\Theta(z) \right)' = 0,
\]
hence
\begin{equation}\label{eq:p_of_z}
p(z) = 2z \left(\ln\Theta(z)\right)'.
\end{equation}
Since $\Theta(z) = 1 + 2z + O(z^2)$, then $2z \left(\ln\Theta(z)\right)' = 4z + O(z^2)$.
Hence, there exists a unique compositional inverse $z(p)$ of $p(z)$ and $z(p)$ is a formal power series (see~\cite[Th. 1.8]{Lando}).
Moreover, $z(p) = \frac{p}{4} + O(p^2)$.

Due to~\cite[Eq. 8.24]{HeilmannLieb}, $\ln\Theta(z)$ is a convex function in $\ln z$ for all $z \in \mathbb{R}_{+}$.
Then equation~\eqref{eq:legendre_first} can be seen as a fact that $-f_2(p)$ as a function of $\frac{p}{2}$
is a Legendre transform of $\ln\Theta(z)$ as a function of $\ln z$ (for introduction to Legendre transform, see~\cite{ZiaRedishMcKay} and~\cite[\S 14]{Arnold}).
Involutivity of Legendre transform (see~\cite[\S 14.C]{Arnold}) implies that
\begin{equation}
\ln\Theta(z) = \sup\limits_{p \in \mathbb{R}} \left\{ \frac{p}{2} \ln z + f_2(p) \right\},
\end{equation}
and the supremum on the right-hand side is reached at $p = p(z)$.
On the other hand, we can find this $p$ also by differentiation
\[
\frac{\operatorname{d}}{\operatorname{dp}} \left( \frac{p}{2} \ln z + f_2(p) \right) = \frac{\ln z}{2} + f_2'(p) = 0.
\]
Hence, $f_2'(p(z)) = -\frac{\ln z}{2}$. Substituting $z(p)$ into $z$, we obtain $f_2'(p) = -\frac{1}{2} \ln z(p)$.
Integrating with respect to $p$, and using the initial condition $f_2(0) = 0$, we conclude that
\begin{equation}\label{eq:f_via_theta}
	f_2(p) = \frac{1}{2}\int\limits_0^p z(p) \operatorname{dp}, \text{ where } p(z) = 2z \left( \ln\Theta(z) \right)'.
\end{equation}
The same formula is obtained in~\cite[IV.A]{Kong2} using another argument.

Finally, we deduce expansion~\eqref{eq:upper_bound} from~\eqref{eq:f_via_theta}.
Using $z(p) = \frac{p}{4} + O(p^2)$, we conclude that
\[
f_2(p) = \frac{1}{2}\int\limits_{0}^p \ln\left( \frac{p}{4} + O(p^2) \right)\operatorname{dp} = \frac{1}{2}\int\limits_{0}^p \left( \ln p - 2\ln 2 + O(p) \right) \operatorname{dp} = \frac{1}{2}\ln p - \frac{1}{2}p - p \ln 2 + \sum\limits_{k = 2}^{\infty} a_k p^k. 
\]
The latter expression is exactly of the same form as the right-hand side in~\eqref{eq:upper_bound}.


\section{Computation of $\Theta(z)$ using $\Theta_{m, n}(z)$}\label{sec:combinatorial}

The goal of the present section is to prove the following theorem, which provides a way to compute the thermodynamical limit $\ln\Theta(z)$.
\begin{theorem}\label{th:main_combinatorial}
	For every integer $N \geqslant 4$,
    \begin{equation}\label{eq:main_th}
    \ln\Theta(z) - \left( S_N - 3 S_{N - 1} + 3 S_{N - 2} - S_{N - 3} \right) = O(z^{N - 1}),
    \end{equation}
    where $S_M = \sum\limits_{m + n = M} \ln \Theta_{m, n}(z)$.
\end{theorem}

In what follows, we will use some properties of the Mayer expansion following~\cite[Section~2.2]{ScottSokal}.

Let $R_{\infty, \infty}$ be the first quadrant of the plane.
We denote the $m \times n$ rectangle whose lower-left corner is the origin by $R_{m, n}$. 
By definition, $R_{m, n} \subset R_{\infty, \infty}$ for all $m$ and $n$.
We denote the set of all dimers in $R_{m, n}$ by $D_{m, n}$.
By definition, the cardinality of $D_{m, n}$ is $2mn - m - n$, and $D_{m, n} \subset D_{\infty, \infty}$ for every $m, n \in \mathbb{Z}_{> 0}$.
For $d_1, d_2 \in D_{\infty, \infty}$, we introduce $W(d_1, d_2)$ by
\[
W(d_1, d_2) = 
\begin{cases}
	1, \text{if } d_1 \text{ and } d_2 \text{ do not overlap,}\\
    0, \text{if } d_1 \text{ and } d_2 \text{ overlap.}
\end{cases}
\]
Using this notation, the grand-canonical partition function introduced in~\eqref{eq:grand_canonical_finite} can be written as (see formula~(1.1a) in~\cite{ScottSokal})
\begin{equation}\label{eq:grand_canonical_W}
  \Theta_{m, n}(z) = \sum\limits_{s = 0}^{\infty} \frac{z^s}{s!} \sum\limits_{(d_1, \ldots, d_s) \in D_{m, n}^s} \left( \prod\limits_{1 \leqslant i < j \leqslant s} W(d_i, d_j) \right),
\end{equation}
where $D_{m, n}^s$ stands for the set of all ordered $s$-tuples of elements of $D_{m, n}$.
Unlike~\eqref{eq:grand_canonical_finite}, formula~\eqref{eq:grand_canonical_W} includes an infinite sum.
However, since among every $\lfloor mn / 2 \rfloor + 1$ dimers there exists at least one pair of overlapping dimers, all terms with $s > \lfloor mn / 2 \rfloor$ vanish.
We introduce $F(d_1, d_2) = W(d_1, d_2) - 1$ for every $d_1, d_2 \in D_{\infty, \infty}$.
Then~\eqref{eq:grand_canonical_W} can be rewritten as (see formula~(2.7) in~\cite{ScottSokal})
\begin{equation}\label{eq:grand_canonical_F}
	\Theta_{m, n}(z) = \sum\limits_{s = 0}^{\infty} \frac{z^s}{s!} \sum\limits_{\mathbf{d} = (d_1, \ldots, d_s) \in D_{m, n}^s} \sum\limits_{G \in \mathcal{G}_s} F(\mathbf{d}, G), \text{ where } F(\mathbf{d}, G) = \prod\limits_{(ij) \in E(G)} F(d_i, d_j),
\end{equation}
and $\mathcal{G}_s$ denotes the set of all graphs on $\{ 1, \ldots, s \}$, and $E(G)$ is the set of edges of a graph $G$.
Changing the order of summation, we obtain $\Theta_{m, n}(z) = \sum\limits_{s = 0}^{\infty} \frac{z^s}{s!} \sum\limits_{G \in \mathcal{G}_s} \mathcal{W}_{m, n}(G)$, where
\begin{equation}\label{eq:defintion_mathcal_W}
  \mathcal{W}_{m, n}(G) = \sum\limits_{\mathbf{d} \in D_{m, n}} F(\mathbf{d}, G).
\end{equation}
In~\cite[p. 1161]{ScottSokal} it is shown that (see formula~(2.11a))
\begin{equation}\label{eq:log_mayer}
  \ln \Theta_{m, n}(z) = \sum\limits_{s = 0}^{\infty} \frac{z^s}{s!} \sum\limits_{G \in \mathcal{C}_s} \mathcal{W}_{m, n}(G), 
\end{equation}
where $\mathcal{C}_s$ is the set of all connected graphs on $\{ 1, \ldots, s \}$.
The above calculations work in quite general context and do not exploit the structure of $D_{m, n}$.
Now we will perform a more careful analysis of~\eqref{eq:defintion_mathcal_W} and~\eqref{eq:log_mayer} in our setting.

For a tuple $\mathbf{d} = (d_1, \ldots, d_s) \in (D_{m, n})^s$, we construct a graph with vertices labeled $1, \ldots, s$ such that there is an edge between $i$ and $j$ if and only if $d_i$ and $d_j$ overlap.
We call a tuple $\mathbf{d}$ \emph{connected} if the corresponding graph is connected.
The set of connected tuples in $D_{m, n}$ of length $s$ is denoted by $\left( D_{m, n} \right)^s_c$.
For $\mathbf{d} \in \left( D_{m, n} \right)^s_c$, we define \emph{the height} (resp., \emph{the width}) of $\mathbf{d}$ to be the number of rows (resp., columns) having nontrivial intersection with at least one of the dimers in $\mathbf{d}$.
We denote it by $h(\mathbf{d})$ (resp., $w(\mathbf{d})$).
Two tuples $\mathbf{d}_1 = (d_1^1, \ldots, d_s^1)$ and $\mathbf{d}_2 = (d_1^2, \ldots, d_s^2)$ are said to be \emph{translation-equivalent} if there exists a translation $\pi$ of the plane by some vector such that $\pi(d_i^1) = d_i^2$ for every $1 \leqslant i \leqslant s$.
This is an equivalence relation, and we write it as $\mathbf{d}_1 \sim \mathbf{d}_2$.

The following facts follow straightforwardly from the definitions
\begin{lemma}\label{lem:tuples_properties}
	\begin{enumerate}
		\item[(i)] For every tuple $(d_1, \ldots, d_s) \in \left(D_{m, n}\right)^s \setminus \left(D_{m, n}\right)^s_c$, the corresponding summand in~\eqref{eq:defintion_mathcal_W} vanishes.
    	\item[(ii)] If $\mathbf{d}_1 = (d_1^1, \ldots, d_s^1)$ and $\mathbf{d}_2 = (d_1^2, \ldots, d_s^2)$ are translation-equivalent, then $F(\mathbf{d}_1, G) = F(\mathbf{d}_2, G)$ for every graph $G \in \mathcal{G}_s$.
    	\item[(iii)] For every connected tuple $\mathbf{d} \in \left( D_{\infty, \infty} \right)^s_c$, the number of tuples $\mathbf{d}' \in \left( D_{m, n} \right)^s_c$ such that $\mathbf{d} \sim \mathbf{d}'$ is exactly
    	\[
    		\left( m - h(\mathbf{d}) + 1 \right)_{+} \cdot \left( n - w(\mathbf{d}) + 1 \right)_{+},
    	\]
        where $(x)_{+} := \max(x, 0)$.
	\end{enumerate}
\end{lemma}
We denote by $\mathcal{T}_s$ a set of tuples in $\left(D_{\infty, \infty}\right)^s_c$ that contains exactly one representative of every equivalence class of
translation-equivalent connected tuples.
Due to Lemma~\ref{lem:tuples_properties} we can rewrite~\eqref{eq:defintion_mathcal_W} as
\begin{align}
\label{eq:W_rewrite}
  \mathcal{W}_{m, n}(G) &= \sum\limits_{\mathbf{d} \in D_{m, n}} F(\mathbf{d}, G)
  \stackrel{(i)}{=} \sum\limits_{\mathbf{d} \in \left(D_{m, n}\right)^s_c} F(\mathbf{d}, G) \\
  &\stackrel{(ii)}{=} \sum\limits_{\mathbf{d} \in \mathcal{T}_s} \left( \sum\limits_{\mathbf{d}' \in D_{m, n}^s, \mathbf{d}' \sim \mathbf{d}} F(\mathbf{d}', G) \right) \nonumber \\
  &\stackrel{(iii)}{=} \sum\limits_{\mathbf{d} \in \mathcal{T}_s} \left( m - h(T) + 1 \right)_{+} \cdot \left( n - w(T) + 1 \right)_{+} \cdot F(\mathbf{d}, G).\nonumber
\end{align}

For $\mathbf{d} \in \left( D_{m, n} \right)_c^s$, we define $\mathcal{W}(\mathbf{d}) = \sum\limits_{G \in \mathcal{C}_s} F(\mathbf{d}, G)$.
Using this notation and~\eqref{eq:W_rewrite}, we can rewrite~\eqref{eq:log_mayer} as
\begin{align*}
  &\ln \Theta_{m, n}(z) = \sum\limits_{s = 0}^{\infty} \frac{z^s}{s!} \sum\limits_{G \in \mathcal{C}_s} \mathcal{W}_{m, n}(G) \\ 
  &= \sum\limits_{s = 0}^{\infty} \frac{z^s}{s!} \sum\limits_{G \in \mathcal{C}_s} \left( \sum\limits_{\mathbf{d} \in \mathcal{T}_s} \left( m - h(\mathbf{d}) + 1 \right)_{+} \cdot \left( n - w(\mathbf{d}) + 1 \right)_{+} \cdot F(\mathbf{d}, G) \right) \\ 
  &= \sum\limits_{s = 0}^{\infty} \frac{z^s}{s!} \sum\limits_{\mathbf{d} \in \mathcal{T}_s} \left( m - h(\mathbf{d}) + 1 \right)_{+} \cdot \left( n - w(\mathbf{d}) + 1 \right)_{+} \cdot \mathcal{W}(\mathbf{d}).
\end{align*}

Hence
\begin{equation}\label{eq:log_t}
  \ln \Theta_{m, n}(z) = \sum\limits_{s = 0}^{\infty} \frac{z^s}{s!} \sum\limits_{\mathbf{d} \in \mathcal{T}_s} \left( m - h(\mathbf{d}) + 1 \right)_{+} \cdot \left( n - w(\mathbf{d}) + 1 \right)_{+} \cdot \mathcal{W}(\mathbf{d}).
\end{equation}

Now we want to obtain a similar expression for $\ln\Theta(z)$ defined in~\eqref{eq:theta_def}
\[
  \ln\Theta(z) = \lim\limits_{m, n \to \infty} \frac{\ln\Theta_{m, n}(z)}{mn} = \sum\limits_{s = 0}^{\infty} \frac{z^s}{s!} \sum\limits_{\mathbf{d} \in \mathcal{T}_s} \lim\limits_{m, n \to \infty} \left( \frac{ \left( m - h(\mathbf{d}) + 1 \right)_{+} \left( n - w(\mathbf{d}) + 1 \right)_{+} }{mn} \right) \mathcal{W}(\mathbf{d}).
\]
Since $\lim\limits_{m, n \to \infty} \frac{ \left( m - h(\mathbf{d}) + 1 \right)_{+} \cdot \left( n - w(\mathbf{d}) + 1 \right)_{+} }{mn} = 1$, we obtain
\begin{equation}\label{eq:log_t_lim}
  \ln\Theta(z) = \sum\limits_{s = 0}^\infty \frac{z^s}{s!} \sum\limits_{\mathbf{d} \in \mathcal{T}_s} \mathcal{W}(\mathbf{d}).
\end{equation}

We are now ready to deduce Theorem~\ref{th:main_combinatorial} from~\eqref{eq:log_t} and~\eqref{eq:log_t_lim}.
\begin{lemma}\label{lem:S_N}
  For every $N \in \mathbb{Z}_{> 0}$
  \[
  S_N = \sum\limits_{s = 0}^{\infty} \frac{z^s}{s!} \sum\limits_{\mathbf{d} \in \mathcal{T}_s} \binom{ (N - w(\mathbf{d}) - h(\mathbf{d}) + 3)_{+} }{3} \mathcal{W}(\mathbf{d}).
  \]
\end{lemma}

\begin{proof}
	By Lemma~\ref{lem:tuples_properties}, the coefficient of $\frac{z^s}{s!}\mathcal{W}(\mathbf{d})$ is equal to
    \[
    \sum\limits_{m + n = N}(m - h(\mathbf{d}) + 1)_{+} \cdot (n - w(\mathbf{d}) + 1)_{+}.
    \]
    If $N < p := w(\mathbf{d}) + h(\mathbf{d})$, the above expression is equal to $0 = \binom{ (N - w(\mathbf{d}) - h(\mathbf{d}) + 3)_{+} }{3}$.
    Otherwise, it is equal to
    \[
    \sum\limits_{k = 1}^{N - p + 1} k \left(N - p + 2 - k \right) = \left(N - p + 2 \right) \left( \sum\limits_{k = 1}^{N - p + 1} k \right) - \left(\sum\limits_{k = 1}^{N - p + 1} k^2\right).
    \]
    It can be verified by direct computation using the formula for the sum of squares that the latter expression is equal to $\binom{N - w(\mathbf{d}) - h(\mathbf{d}) + 3}{3}$.
    This proves the lemma.
\end{proof}

Fix some $s \leqslant N - 2$ and $\mathbf{d} \in \mathcal{T}_s$.
We will prove that all summands of the form $\frac{z^s}{s!}\mathcal{W}(\mathbf{d})$ on the left-hand side of~\eqref{eq:main_th} cancel.
Since $\mathbf{d}$ is connected, it contains at least $w(\mathbf{d}) - 1$ horizontal dimers and at least $h(\mathbf{d}) - 1$ vertical dimers.
Hence, $w(\mathbf{d}) + h(\mathbf{d}) - 2 \leqslant s \leqslant N - 2$, so $p := w(\mathbf{d}) + h(\mathbf{d}) \leqslant N$.
This inequality together with Lemma~\ref{lem:S_N} and~\eqref{eq:log_t_lim} implies that the coefficient of $\frac{z^s}{s!}\mathcal{W}(\mathbf{d})$ on the left-hand side of~\eqref{eq:main_th} is equal to
\[
1 - \left( \binom{N - p + 3}{3} - 3\binom{N - p + 2}{3} + 3\binom{N - p + 1}{3} - \binom{N - p}{3} \right).
\]
Expanding the brackets, we verify that this expression is zero for every $N - p \geqslant 0$.
This concludes the proof of Theorem~\ref{th:main_combinatorial}.


\section{Description of the algorithm}\label{sec:algorithm}

\subsection{General algorithm.} 
Combining~\eqref{eq:f_via_theta} and Theorem~\ref{th:main_combinatorial}, we obtain Algorithm~\ref{alg:draft}, 
the first version of an algorithm for computing the first $n$ terms of $f_2(p)$.
Note that
\begin{itemize}
	\item line~\ref{line:compute_Theta} is correct due to Theorem~\ref{th:main_combinatorial};
    \item line~\ref{line:compute_f2} is correct due to~\eqref{eq:f_via_theta};
    \item procedure $\operatorname{ComputeTheta}$ is described in subsection~\ref{subsec:compute_theta};
    \item procedure $\operatorname{InversePowerSeries}(a(z))$ computes a power series $z(p)$ given a power series $p(z)$ (see~\cite[Th. 1.8]{Lando}).
\end{itemize}

\begin{algorithm}
\caption{Nonoptimized version of the algorithm}\label{alg:draft}
	\KwIn{Nonnegative integer $n$.}
    \KwOut{$f_2(p)$ modulo $O(p^n)$.}

  \For{$i$ from $1$ to $\left\lfloor \frac{n + 1}{2} \right\rfloor$}{
      $[\Theta_{i, 1}(z), \ldots, \Theta_{i, n_0 + 1 - i}(z)] := \operatorname{ComputeTheta}(i, n_0 + 1 - i);$\\
      \For{$j$ from $1$ to $n_0 + 1 - i$}{
        $\Theta_{j, i}(z) := \Theta_{i, j}(z);$
      }
  }
  
  \For{$k$ from $0$ to $3$}{
  	$S_{n + 1 - k} := \sum\limits_{i + j = n + 1 - k} \ln\Theta_{i, j} (z);$
  }
  
  $\ln\Theta(z) := S_{n + 1} - 3S_n + 3S_{n - 1} - S_{n - 2};$\label{line:compute_Theta}\\
  $z(p) := \operatorname{InversePowerSeries}\left(2z \left( \ln\Theta(z) \right)' \right);$\\
  $f_2(p) := -\frac{1}{2} \int \ln z(p) \DP;$\label{line:compute_f2}\\
  \Return $f_2(p)$;
\end{algorithm}

Several improvements can be made:
\begin{itemize}
	\item Computation of $\operatorname{ComputeTheta}(i, j)$ deals with a very long vector of possibly very large numbers (see~Section~\ref{subsec:compute_theta}).
    In order to fit into the memory, we perform computation modulo several primes and use chinese remaindering and rational reconstruction to obtain the final result (see~Section~\ref{subsec:rational_reconstruction}).
    
    \item The output of Algorithm~\ref{alg:draft} with input $n$, let us call it $\tilde{g}_n(p)$, coincides with $f_2(p)$ only modulo $O(z^n)$.
    Nevertheless, the first few nonzero coefficients of $f_2(p) - \tilde{g}_n(p)$ turn out to satisfy linear recurrence relations with respect to $n$, so
   they can be computed easily. This allows us to ``correct'' these terms and obtain a more precise result. See Section~\ref{subsec:correction} for further details.
    
    \item Since we need only the first $n$ terms of $\Theta(z)$, it is sufficient to compute only the first $n$ terms for every computed $\Theta_{i, j}(z)$.
    Therefore, all intermediate polynomials can also be truncated.
\end{itemize}

With these improvements, we obtain the final version of our algorithm.
For more details, see the source code (see Section~\ref{subsec:implement}).

  
  


\subsection{Computation of $\Theta_{m, n}(z)$.}\label{subsec:compute_theta}
We will compute $\Theta_{m, n}(z)$ using an optimization of \textit{the transfer--matrix method} (see~\cite[\S 4.7]{StanleyVol1}).
Fix a positive integer $m$.
Let $n$ be a nonnegative integer, and $0 \leqslant N < 2^m$.
Viewing $N$ as a vector of $m$ bits, we denote the $i$-th bit of $N$ by $N[i]$.
We denote by $F_N^{(m, n)}$ the polygon obtained form the $m \times n$ rectangle by adding one additional cell (we will call it \emph{an external} cell) to the
end of every row such that $N[i] = 1$, where $i$ is the index of the row.
For example, $F_5^{(4, 6)}$ is shown on Figure~\ref{fig:F5_6}. In particular, $F_0^{(m, n + 1)}$ is the same as $F_{2^m - 1}^{(m, n)}$.

\begin{figure}
\begin{tikzpicture}
  \node[] at (-0.2, 1.75) {\textbf{1}};
  \node[] at (-0.2, 1.25) {\textbf{2}};
  \node[] at (-0.2, 0.75) {\textbf{3}};
  \node[] at (-0.2, 0.25) {\textbf{4}};
  \draw[step=0.5cm,gray,very thin,dotted] (3, 0) grid (3.5,2);
  \draw[step=0.5cm] (0, 0) grid (3,2);
  \draw[step=0.5cm] (3, 0) grid (3.5,0.5);
  \draw[step=0.5cm] (3, 0.99) grid (3.5,1.5);
\end{tikzpicture}
\caption{$F_5^{(6)}$}\label{fig:F5_6}
\end{figure}

We introduce polynomial $P_N^{(m, n)} (z)$ to be a generating function for the number of tilings of $F_N^{(m, n)}$ 
such that every external cell is covered by a horizontal dimer, i.e. 
$P_N^{(m, n)}(z) = \sum\limits_{j = 0}^{m(n + 1)} a_{N, j}^{(m, n)} z^j$,
where $a_{N, j}^{(m, n)}$ is the number of monomer-dimer tilings of $F_N^{(m, n)}$ with exactly $j$ dimers such that every external cell is covered by a horizontal dimer.
We will call such tilings \textit{rigid}.
In particular, $\Theta_{m, n}(z) = P_{0}^{(m, n)}(z)$.
We denote by $P^{(m, n)}$ the vector $(P_0^{(m, n)}(z), \ldots, P_{2^m - 1}^{(m, n)}(z))$.

\begin{remark}
  It can be shown (using techniques from~\cite[\S V.6.]{FlajoletSedgewick}), 
  that there exists a matrix $M$ with entries in $\mathbb{Z}[z]$ such that $P^{(m, n + 1)} = MP^{(m, n)}$.
  Hence, $\Theta_{m, n}(z)$ can be computed as the first coordinate of $M^n P^{(m, 0)}$.
  However, in our computations $m$ can be any natural number up to $30$, 
  so $M$ can have $2^{30} \times 2^{30} = 2^{60} \approx 10^{18}$ entries.
  Luckily, the matrix $M$ is highly structured (see~\cite{Lieb}), so there exists
  a faster algorithm for computing $P^{(m, n + 1)}$ from $P^{(m, n)}$.
\end{remark}

We present an algorithm (Algorithm~\ref{alg:multiplication_transfer}) that computes $P^{(m, n + 1)}$ from $P^{(m, n)}$ in-place 
(i.e. with $O(1)$ additional space) using $O(m2^m)$ arithmetic operations.
We denote the number of ones in the binary representation of $N$ by $\operatorname{BinDig}(N)$.

\begin{algorithm}
\caption{Computing $P^{(m, n + 1)}$ from $P^{(m, n)}$.}\label{alg:multiplication_transfer}
	\KwIn{Vector $P = (P_{0}^{(m, n)}(z), \ldots, P_{2^m - 1}^{(m, n)}(z))$.}
    \KwOut{Vector $P = (P_{0}^{(m, n + 1)}(z), \ldots, P_{2^m - 1}^{(m, n + 1)}(z))$.}

    \For{$N$ from $0$ to $2^{m - 1} - 1$\label{line:reverse}}{
      Swap values $P[N]$ and $P[2^m - 1 - N];$
    }
    
    \For{$j$ from $1$ to $m$\label{line:loop_bits}}{
      \For{$N$ from $0$ to $2^m - 1$}{
        \If{$N[j] = 0$}{
          $P[N] \mathrel{+}= P[N + 2^{m - j}];$\label{line:monomer}
        }
        \If{$N[j] = 0$ and $j > 1$ and $N[j - 1] = 0$}{
          $P[N] \mathrel{+}= z \cdot P[N + 2^{m - j} + 2^{m - j + 1}];$\label{line:vert_dimer}
        }
      }
      }
      \For{$N$ from $0$ to $2^m - 1$}{
        $d:= \operatorname{BinDig}(N);$\\
        $P[N] := z^d \cdot P[N];$\label{line:mult_by_bd}\\
      }
      \Return $P;$
\end{algorithm}

\begin{proposition}
	Algorithm~\ref{alg:multiplication_transfer} is correct.
\end{proposition}

\begin{proof}
  We will prove by induction on $j$ that after the $j$-th iteration of the loop in line~\ref{line:loop_bits} (for $j = 0$ it means the moment just before the first iteration)
  $\widetilde{P}_N := z^{\operatorname{BinDig}(N)} \cdot P[N]$ is the generating polynomial for the number of monomer-dimer tilings of $F_N^{(m, n)}$ satisfying the following
  $A_j$ property:
  
  \textit{\underline{$A_j$ property:} the tiling is rigid, and the right-most cell in rows with the number greater than $j$ is covered by a horizontal dimer.
  }
  
  First we prove the base case, where $j = 0$. Due to the loop in line~\ref{line:reverse}, $P[N] = P_{2^m - 1 - N}^{(m, n)}(z)$.
  Since the binary representation of $2^m - 1 - N$ can be obtained from the binary representation of $N$ by inverting all $m$ bits,
  adding a horizontal dimer to the end of every row without an external cell provides us a bijection between the set of rigid tilings of $F_{2^m - 1 - N}^{(m, n)}$
  and the set of tilings of $F_N^{(m, n + 1)}$ with $A_0$ property (see Fig.~\ref{fig:induction_base}).
  This map adds $\operatorname{BinDig}(N)$ new dimers, so the corresponding generating polynomials differ by the factor $z^{\operatorname{BinDig}(N)}$.

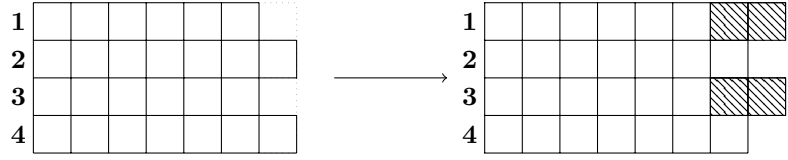
\begin{figure}[h!]
\begin{tikzpicture}
  \node[] at (-0.2, 1.75) {\textbf{1}};
  \node[] at (-0.2, 1.25) {\textbf{2}};
  \node[] at (-0.2, 0.75) {\textbf{3}};
  \node[] at (-0.2, 0.25) {\textbf{4}};
  \draw[step=0.5cm,gray,very thin,dotted] (3, 0) grid (3.5,2);
  \draw[step=0.5cm] (0, 0) grid (3,2);
  \draw[step=0.5cm] (3, 0) grid (3.5,0.5);
  \draw[step=0.5cm] (3, 0.99) grid (3.5,1.5);
  
  \draw[->] (4, 1) -- (\shift - 0.5,1);
  
  \node[] at (-0.2 + \shift, 1.75) {\textbf{1}};
  \node[] at (-0.2 + \shift, 1.25) {\textbf{2}};
  \node[] at (-0.2 + \shift, 0.75) {\textbf{3}};
  \node[] at (-0.2 + \shift, 0.25) {\textbf{4}};
  \draw[step=0.5cm] (-0.01 + \shift, 0) grid (3 + \shift,2);
  \draw[step=0.5cm] (3 + \shift, 0) grid (3.5 + \shift,0.5);
  \draw[step=0.5cm] (3 + \shift, 0.99) grid (3.5 + \shift,1.5);
  \draw[step=0.5cm] (3 + \shift, 0.5) grid (4 + \shift,1);
  \draw[pattern=north west lines] (3 + \shift, 0.5) rectangle (3.5 + \shift,1);
  \draw[pattern=north west lines] (3.5 + \shift, 0.5) rectangle (4 + \shift,1);
  \draw[step=0.5cm] (3 + \shift, 1.5) grid (4 + \shift,2);
  \draw[pattern=north west lines] (3 + \shift, 1.5) rectangle (3.5 + \shift,2);
  \draw[pattern=north west lines] (3.5 + \shift, 1.5) rectangle (4 + \shift,2);
\end{tikzpicture}
\caption{Tilings of $F_5^{(6)}$ to tilings of $F_{10}^{(7)}$ with property $A_0$}\label{fig:induction_base}
\end{figure}

  Assume now that $j > 0$. For $N$ such that $N[j] = 1$, properties $A_{j - 1}$ and $A_{j}$ are the same, so the corresponding component of vector $P$
  should not be changed. Assume that $N[j] = 0$. We denote the last cell of the $j$-th row in $F_N^{(m, n + 1)}$ by $c$.
  Consider an arbitrary monomer-dimer tilings of $F_N^{(m, n + 1)}$ with property $A_j$.
  There are three options for $c$:
  \begin{enumerate}
  	\item Cell $c$ is covered by a horizontal dimer. Then this tiling has also property $A_{j - 1}$ and is already counted in $z^{\operatorname{BinDig}(N)} P[N]$.
    
    \item Cell $c$ is covered by a monomer. Replacing this monomer by a horizontal dimer, we establish a bijection between such tilings of $F_N^{(m, n)}$ and tilings of $F_{N + 2^{m - j}}^{(m, n + 1)}$ with property $A_{j - 1}$ (see Figure~\ref{fig:induction_c_monomer}). 
    Due to the induction hypothesis, the generating polynomial for the latter is $\widetilde{P}_{N + 2^{m - j}}$.
    Hence, in order to take into account tilings where $c$ is covered by a monomer, we should add $\frac{1}{z} \widetilde{P}_{N + 2^{m - j}}$ to $\widetilde{P}_N$.
    This is equivalent to $P[N] \mathrel{+}= P[N + 2^{m - j}]$ in  line~\ref{line:monomer}.
    
\begin{figure}[h!]
\begin{tikzpicture}  
  \node[] at (-0.2, 1.75) {\textbf{1}};
  \node[] at (-0.2, 1.25) {\textbf{2}};
  \node[] at (-0.2, 0.75) {\textbf{3}};
  \node[] at (-0.2, 0.25) {\textbf{4}};
  \draw[step=0.5cm] (-0.01, 0) grid (3.5,2);
  \draw[pattern=north west lines] (3, 1.5) rectangle (3.5,2);
  \draw[pattern=north west lines] (3.5, 1.5) rectangle (4,2);
  \draw[line width=1mm] (3, 0.5) rectangle (3.5, 1);
  \node[] at (3.25,0.75) {\textbf{c}};
  
  \draw[->] (4.5, 1) -- (\shift - 0.5,1);

  \node[] at (-0.2 + \shift, 1.75) {\textbf{1}};
  \node[] at (-0.2 + \shift, 1.25) {\textbf{2}};
  \node[] at (-0.2 + \shift, 0.75) {\textbf{3}};
  \node[] at (-0.2 + \shift, 0.25) {\textbf{4}};
  \draw[step=0.5cm] (-0.01 + \shift, 0) grid (3 + \shift,2);
  \draw[step=0.5cm] (3 + \shift, 0) grid (3.5 + \shift,0.5);
  \draw[step=0.5cm] (3 + \shift, 0.99) grid (3.5 + \shift,1.5);
  \draw[step=0.5cm] (3 + \shift, 1.5) grid (4 + \shift,2);
  \draw[pattern=north west lines] (3 + \shift, 1.5) rectangle (3.5 + \shift,2);
  \draw[pattern=north west lines] (3.5 + \shift, 1.5) rectangle (4 + \shift,2);
  \draw[step=0.5cm] (3 + \shift, 0.5) grid (4 + \shift,1);
  \draw[pattern=north west lines] (3 + \shift, 0.5) rectangle (3.5 + \shift,1);
  \draw[pattern=north west lines] (3.5 + \shift, 0.5) rectangle (4 + \shift,1);
  \draw[line width=1mm] (3 + \shift, 0.5) rectangle (4 + \shift, 1);
\end{tikzpicture}
\caption{Tilings of $F_{8}^{(7)}$ with property $A_3$ to tilings of $F_{10}^{(7)}$ with property $A_2$}\label{fig:induction_c_monomer}
\end{figure}
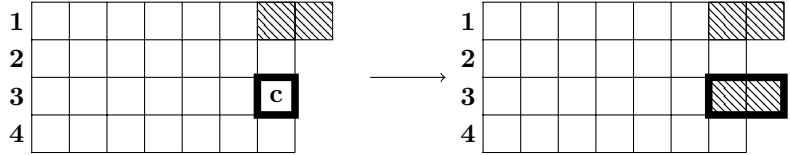
    
    \item Cell $c$ is covered by a vertical dimer. This dimer cannot cover also the cell below $c$ due to $A_{j}$ property.
    Hence, it covers $c$ and the cell above, say $d$, so $N[j - 1] = 0$.
    Replacing this dimer with two horizontal dimers, we establish a bijection between such tilings of $F_N^{(m, n)}$ and tilings of $F_{N + 2^{m - j} + 2^{m - j + 1}}^{(m, n + 1)}$ with property $A_{j - 1}$.
    These cases are counted in line~\ref{line:vert_dimer}.
    
\begin{figure}[h!]
\begin{tikzpicture}  
  \node[] at (-0.2, 1.75) {\textbf{1}};
  \node[] at (-0.2, 1.25) {\textbf{2}};
  \node[] at (-0.2, 0.75) {\textbf{3}};
  \node[] at (-0.2, 0.25) {\textbf{4}};
  \draw[step=0.5cm] (-0.01, 0) grid (3.5,2);
  \draw[pattern=north west lines] (3, 1.5) rectangle (3.5,2);
  \draw[pattern=north west lines] (3.5, 1.5) rectangle (4,2);
  \draw[line width=1mm] (3, 0.5) rectangle (3.5, 1.5);
  \node[] at (3.25,0.75) {\textbf{c}};
  \node[] at (3.25,1.25) {\textbf{d}};

  \draw[->] (4.5, 1) -- (\shift - 0.5,1);

  \node[] at (-0.2 + \shift, 1.75) {\textbf{1}};
  \node[] at (-0.2 + \shift, 1.25) {\textbf{2}};
  \node[] at (-0.2 + \shift, 0.75) {\textbf{3}};
  \node[] at (-0.2 + \shift, 0.25) {\textbf{4}};
  \draw[step=0.5cm] (-0.01 + \shift, 0) grid (3 + \shift,2);
  \draw[step=0.5cm] (3 + \shift, 0) grid (3.5 + \shift,0.5);
  \draw[step=0.5cm] (3 + \shift, 0.99) grid (3.5 + \shift,1.5);
  \draw[step=0.5cm] (3 + \shift, 0.5) grid (4 + \shift,1);
  \draw[pattern=north west lines] (3 + \shift, 0.5) rectangle (3.5 + \shift,1);
  \draw[pattern=north west lines] (3.5 + \shift, 0.5) rectangle (4 + \shift,1);
  \draw[step=0.5cm] (3 + \shift, 1.5) grid (4 + \shift,2);
  \draw[pattern=north west lines] (3 + \shift, 1.5) rectangle (3.5 + \shift,2);
  \draw[pattern=north west lines] (3.5 + \shift, 1.5) rectangle (4 + \shift,2);
  \draw[step=0.5cm] (3 + \shift, 1) grid (4 + \shift,1.5);
  \draw[pattern=north east lines] (3 + \shift, 1) rectangle (3.5 + \shift,1.5);
  \draw[pattern=north east lines] (3.5 + \shift, 1) rectangle (4 + \shift,1.5);
  \draw[line width=1mm] (3 + \shift, 0.5) rectangle (4 + \shift, 1);
  \draw[line width=1mm] (3 + \shift, 1) rectangle (4 + \shift, 1.5);
\end{tikzpicture}
\caption{Tilings of $F_{8}^{(7)}$ with property $A_3$ to tilings of $F_{14}^{(7)}$ with property $A_2$}\label{fig:induction_c_dimer}
\end{figure}
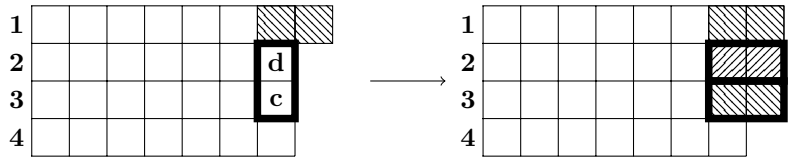

  \end{enumerate}
  
  Since the $A_n$ property is just rigidness, after multiplication by an appropriate degree of $z$ in line~\ref{line:mult_by_bd} we obtain the vector $P^{(m, n + 1)}$.
\end{proof}

\begin{remark}
	Algorithm~\ref{alg:multiplication_transfer} can be parallelized. Consider an iteration of the loop in line~\ref{line:loop_bits}
    with $j > 0$. Then, during the iteration, coordinates of $P[N]$ with different $N[0]$ do not interact, so the whole vector can be 
    divided into two halves (depending on $N[0]$), and these halves can be processed by separate threads.
    Taking into accout $N[1]$, we can divide the work between four threads, and so on.
    In our computation, we used $32$ threads (so, we divided the work based on $N[0], \ldots, N[4]$).
\end{remark}

Finally, using Algorithm~\ref{alg:multiplication_transfer}, we can write a pseudocode for procedure $\operatorname{ComputeTheta}(m, n)$, see Algorithm~\ref{alg:compute_theta}.

\begin{algorithm}
\caption{$\operatorname{ComputeTheta}$}\label{alg:compute_theta}
	\KwIn{Natural numbers $m$ and $n$.}
    \KwOut{Vector of polynomials $[\Theta_{m, 1}(z), \ldots, \Theta_{m, n}(z)]$.}

    $\operatorname{result} := [];$\\
    $P :=$ zero vector of polynomials in $z$ of length $2^m$;\\
    $P[0] : =1;$\\
    \For{$i$ from $1$ to $n$}{
      Apply Algorithm~\ref{alg:multiplication_transfer} to $P$;\\
      Append $P[0]$ to $\operatorname{result}$;
    }
    
    \Return $\operatorname{result}$;
\end{algorithm}


\subsection{Correction terms}\label{subsec:correction}
We can compute more terms of $\Theta(z)$ and, consequently, of $f_2(p)$ if we examine carefully the right-hand side of~\eqref{eq:main_th}. 
Below we write down the first nonzero term of the right-hand side of~\eqref{eq:main_th} for $N = 4, 5, \ldots$
\[
\mathbf{11} z^3, \mathbf{-38} z^4, \mathbf{115} z^5, \mathbf{-309} z^6, \mathbf{759} z^7, \mathbf{-1748} z^8, \mathbf{3847} z^9, \mathbf{-8203} z^{10}, \mathbf{17115} z^{11}, \ldots
\]
Denote the sequence of coefficients by $\{ a_n \}_{n = 1}^{\infty}$. 
Using {\sc Guess} package (\cite{KauersGuess}, for introduction to guessing, see~\cite[\S 4]{KauersToolkit}) we find that this sequence (we computed first $50$ terms) satisfies the following recurrence relation
\begin{equation}\label{eq:recurrence_a}
a_{n + 5} = -6 a_{n + 4} - 14 a_{n + 3} - 16 a_{n + 2} - 9 a_{n + 1} - 2 a_n.
\end{equation}
Using~\eqref{eq:recurrence_a}, we can compute $a_n$ easily, so we get one more correct term of $\Theta(z)$.
Instead of giving a rigorous proof of~\eqref{eq:recurrence_a} which is long and involved, we would like to explain informally why it is natural to expect such a relation.

Formula~\eqref{eq:log_t_lim} shows that the coefficient of $z^s$ in $\Theta(z)$ is a sum of weights of all connected polyominos constructed from $s$ overlapping dimers.
On the other hand, the argument after Lemma~\ref{lem:S_N} shows that the coefficient of $z^s$ in 
\[
S_N - 3 S_{N - 1} + 3 S_{N - 2} - S_{N - 3}
\]
is a sum of weights over all connected polyominos constructed from $s$ overlapping dimers with the sum of height and width at most $N - 2$.
Hence, the coefficient of $z^{N - 1}$ in their difference is a sum of weights of all connected polyominos constructed from $N$ overlapping dimers with the sum of the height and the width exactly $N - 1$ (the sum can not be larger for a connected polyomino).
These requirements on a polyomino are quite restrictive, by a combinatorial argument one can see that all such polyominos are ``of a similar shape'' as those in Figure~\ref{fig:polyominoe}.
More precisely, there exist two cells ($a$ and $b$ in the figure), maybe coinciding, such that
each of them is connected to two sides of an $m \times n$ ($m + n = N - 1$) rectangle by straight lines, and $a$ and $b$ are connected by a path such that at each step the path becomes closer to $b$ (all such paths have the same length).
Counting such polyominos is a standard combinatorial problem (similar counting problems for polyominos are discussed in~\cite[\S 4.7.5]{StanleyVol1}),
that is very likely to result in a formula satisfying a linear recurrence.

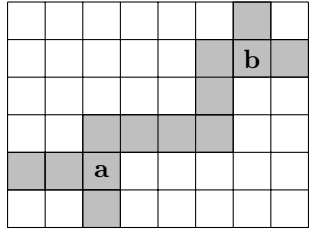
\begin{figure}
\begin{tikzpicture}  
  \draw[fill = lightgray] (0, 0.5) rectangle (0.5,1);
  \draw[fill = lightgray] (0.5, 0.5) rectangle (1,1);
  \draw[fill = lightgray] (1, 0.5) rectangle (1.5,1);
  \draw[fill = lightgray] (1, 0) rectangle (1.5,0.5);
  \draw[fill = lightgray] (1, 1) rectangle (3,1.5);
  \draw[fill = lightgray] (2.5, 1.5) rectangle (3,2.5);
  \draw[fill = lightgray] (3, 2) rectangle (4,2.5);
  \draw[fill = lightgray] (3, 2.5) rectangle (3.5,3);
  \draw[step=0.5cm] (0, 0) grid (4,3);
  \node[] at (1.25,0.75) {\textbf{a}};
  \node[] at (3.25,2.25) {\textbf{b}};
\end{tikzpicture}
\caption{``Large'' and ``thin'' polyomino for $N = 15$}\label{fig:polyominoe}
\end{figure}

Moreover, the same argument shows that there also should be a combinatorial description and a similar recurrence for the second nonzero term in the left-hand side of~\eqref{eq:main_th}, the third, the fourth and so on.
Our data was enough to discover and verify five formulas of this type (from the first until the fifth nonzero term in~\eqref{eq:main_th}).
This is the recurrence for the second nonzero coefficient
\[
b_{n + 7} = -9 b_{n + 6} - 34 b_{n + 5} - 70 b_{n + 4} - 85 b_{n + 3} - 61 b_{n + 2} - 24 b_{n + 1} - 4 b_n.
\]
We omit the others, because they are too large.
However, in our program we do not use recurrences themselves, but the closed form expression for their solutions.
This allows us to compute five more terms of $\Theta(z)$ and, consequently, of $f_2(p)$.



\subsection{Modular computation}\label{subsec:rational_reconstruction} The largest $n$ we used as an input of the algorithm in our computation was $65$.
Taking into account correction terms, this means that $\operatorname{ComputeTheta}$ is invoked with parameters $30$ and $31$.
Hence, the vector $P$ in Algorithm~\ref{alg:multiplication_transfer} will have $2^{30} \approx 10^9$ entries.
Every entry is a polynomial (in our computations it is a truncated polynomial with only $70$ terms), hence in total we have $7.5 \cdot 10^{10}$ integers at every moment.
Since these integers represent the number of tilings of a rectangle, they grow fast, so storing them all exactly would require at least several terabytes of memory.
However, the final result is a list of coefficients of a power series for $f_2(p)$, that is just $65$ rational numbers.
A standard way to deal with such situation (see~\cite[\S 4.2]{KauersToolkit}) is to use computations modulo prime $p$ for intermediate steps.
If $p \leqslant 2^{31} - 1$, then all numbers will fit into $32$ bits, and the whole vector $P$ will occupy just $270$ gigabytes.
Repeating this computation for different primes, we can reconstruct the coefficients of $f_2(p)$ using the chinese remaindering (see~\cite[\S 5.4]{MCA}) and the rational reconstruction procedure (see~\cite[\S 5.10]{MCA}).

The question is how many primes we should take. We start with $2^{31} - 1$, and add new prime numbers until the result of the reconstruction stabilizes.
It turned out that fifteen prime numbers (from~$2^{31} - 1 = 2147483647$ down to~$2147483269$) are enough, however we computed several more in order to make sure that the result is correct.
The correctness of the result is further justified by the comparison in Section~\ref{sec:numerical}.


\section{Numerical results and implementation}\label{sec:numerical}

\subsection{Implementation.}\label{subsec:implement} We implemented most of our algorithm in {\sc Sage} except 
the function $\operatorname{ComputeTheta}$, which was implemented in {\sc C}. See the source code in \url{github.com/pogudingleb/monomer_dimer_tilings}.
Computation modulo one prime with $n = 5$ took about two days using $32$ cores and $270$ gigabytes of memory.
Since we need fifteen primes, the whole computation took about one month.

\subsection{Numerical results.} Table~\ref{table:ak} contains $a_k$'s (defied in~\eqref{eq:lower_bound}) obtain by our computation.
Expanding $(1 - p)\ln (1 - p)$ into Taylor series at $p = 0$, we obtain the following formula expressing 
$b_k$ defined in~\eqref{eq:upper_bound} via $a_k$:
\begin{equation}\label{eq:bk_via_ak}
	b_k = a_k - \frac{1}{k(k - 1)}.
\end{equation}
We introduce following truncated versions of~\eqref{eq:lower_bound} and~\eqref{eq:upper_bound}
\begin{align*}
U_n(p) &= \frac{1}{2} \left( (2\ln 2 + 1)p - p\ln p\right) + \sum\limits_{j = 2}^n b_j p^j,\\
L_n(p) &= \frac{1}{2} \left( (2\ln 2 - 1)p - p\ln p\right) - (1 - p)\ln (1 - p) + \sum\limits_{j = 2}^n a_j p^j.
\end{align*}

All computed $63$ values $a_k$ are positive, all computed $63$ values $b_k$ are negative.
Assuming that this pattern persists, we can write
\[
L_n(p) \leqslant f_2(p) \leqslant U_n(p).
\]
This provides us with lower and upper bound for $f_2(p)$.
We plot both $L_{64}(p)$ and $U_{64}(p)$ together for $p \in [0, 1]$ on Figure~\ref{fig:plot01}. 
Dashed curve on this plot is $-p \ln p - (1 - p) \ln (1 - p)$, 
that is the negative value of the free energy for monomer-monomer problem with two different types of monomers.
We also plot both $L_{64}(p)$ and $U_{64}(p)$ for $p \in [0.9, 1]$ on Figure~\ref{fig:plot091}.

\begin{figure}[h!]
	\centering
	\begin{subfigure}{.5\textwidth}
    \centering
	\includegraphics[width=0.8\textwidth]{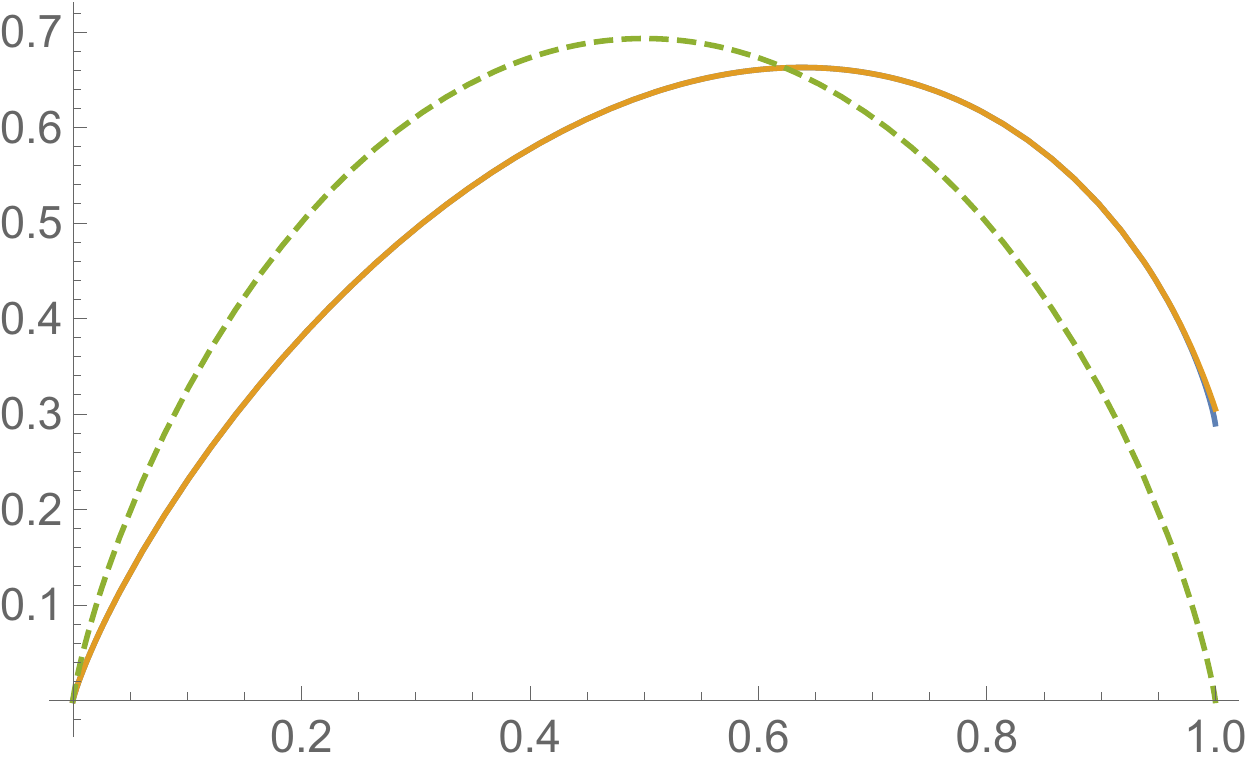}
    \caption{On $[0, 1]$}\label{fig:plot01}
    \end{subfigure}%
    \begin{subfigure}{.5\textwidth}
    \centering
	\includegraphics[width=0.8\textwidth]{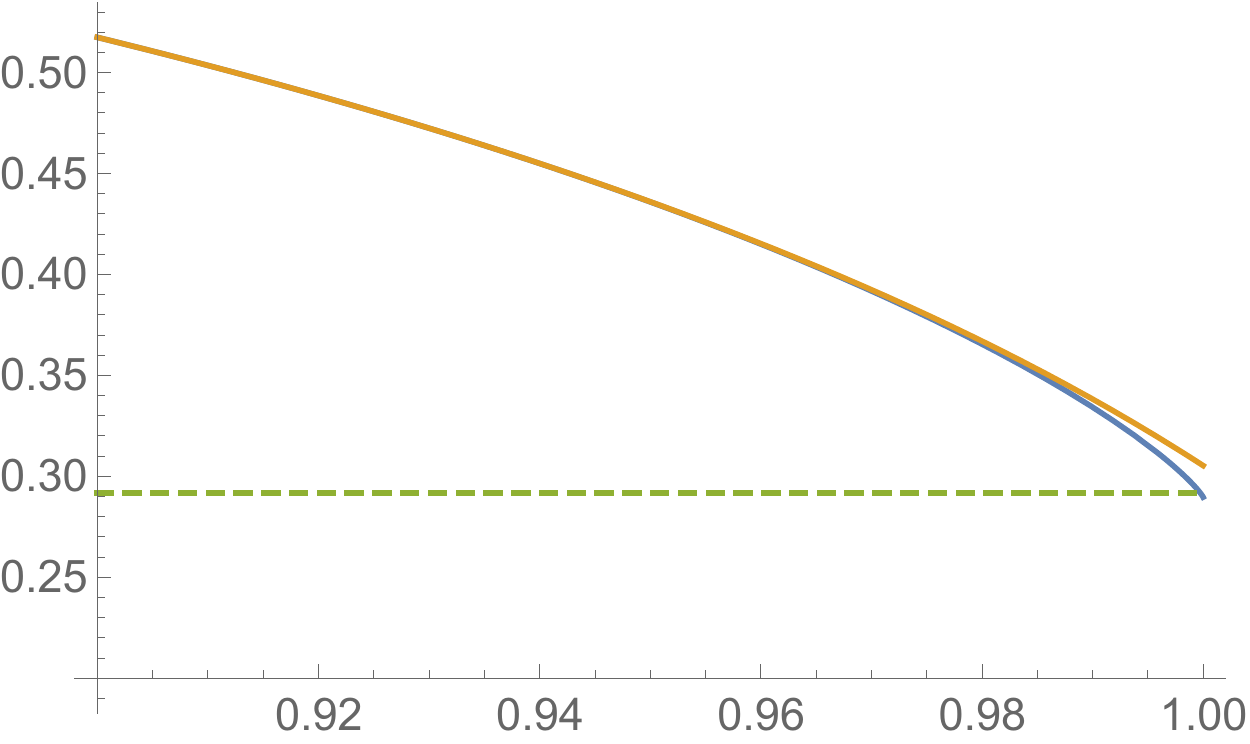}
    \caption{On $[0.9, 1]$, dashed line is $y = f_2(1) = \frac{G}{\pi}$}\label{fig:plot091}
    \end{subfigure}
    \caption{Plots of $L_{64}(p)$ and $U_{64}(p)$}
\end{figure}

Plots of $L_{64}(p)$ and $U_{64}(p)$ in Figure~\ref{fig:plot01} are indistinguishable, the difference between them in Figure~\ref{fig:plot091} is visible only very close to $p = 1$.
On Figure~\ref{fig:plot091} we also see that lower bound is much more accurate at $p = 1$.
The difference $U_{64}(p) - L_{64}(p)$
does not exceed $2.3 \cdot 10^{-16}$ for $p \in [0, 0.5]$ and $2.1\cdot 10^{-6}$ for $p \in [0, 0.9]$.
Note that for $U_{24}(p) - L_{24}(p)$ (these two bounds could be computed using results of~\cite{ButeraFederbushPernici})
 these numbers are~$9.3 \cdot 10^{-11}$ and~$7.5\cdot 10^{-4}$, respectively, so our bound reduces the error by several orders of magnitude.

\begingroup
\renewcommand\arraystretch{1.5}
\begin{longtable}{|c|c|c|c|c|c|}
\caption{Values of $a_k$}\label{table:ak} \\
	\hline
	$k$ & $a_k$ & $k$ & $a_k$ & $k$ & $a_k$ \\ \hline
$ 2 $ & $ \frac{1}{16} $ & $ 23 $ & $ \frac{4312434281365}{17803292276948992} $ & $ 44 $ & $ \frac{18487601206244410582171859}{292772819290992435013642878976} $  \\ \hline
$ 3 $ & $ \frac{1}{192} $ & $ 24 $ & $ \frac{5789230773063}{25895697857380352} $ & $ 45 $ & $ \frac{74150661042096992710148129}{1225560638892526472150132981760} $ \\ \hline
$ 4 $ & $ \frac{7}{1536} $ & $ 25 $ & $ \frac{69044819053441}{337769972052787200} $ & $ 46 $ & $ \frac{297604910587450946018199331}{5125071762641474338082374287360} $ \\ \hline
$ 5 $ & $ \frac{41}{10240} $ & $ 26 $ & $ \frac{272097812497681}{1463669878895411200} $ & $ 47 $ & $ \frac{1194303993371769853836734501}{21411410919479937234655252578304} $ \\ \hline
$ 6 $ & $ \frac{181}{61440} $ & $ 27 $ & $ \frac{1068966474984721}{6323053876828176384} $ & $ 48 $ & $ \frac{4789513328571295127284133845}{89369367316090172805517575979008} $ \\ \hline
$ 7 $ & $ \frac{757}{344064} $ & $ 28 $ & $ \frac{601281977474899}{3891110078048108544} $ & $ 49 $ & $ \frac{19188774086998950351884051009}{372689276467099444040030742380544} $ \\ \hline
$ 8 $ & $ \frac{3291}{1835008} $ & $ 29 $ & $ \frac{16672616519735441}{117021532717594968064} $ & $ 50 $ & $ \frac{76803645872757902332578961121}{1552871985279581016833461426585600} $ \\ \hline
$ 9 $ & $ \frac{14689}{9437184} $ & $ 30 $ & $ \frac{66545602395606901}{501520854503978434560} $ & $ 51 $ & $ \frac{307176141884436645170078617001}{6465018061163969947633186347417600} $ \\ \hline
$ 10 $ & $ \frac{64771}{47185920} $ & $ 31 $ & $ \frac{267471214350929957}{2144433998568735375360} $ & $ 52 $ & $ \frac{1228026136368811312663436458705}{26894475134442114982154055205257216} $ \\ \hline
$ 11 $ & $ \frac{276101}{230686720} $ & $ 32 $ & $ \frac{1080431496491179115}{9149585060559937601536} $ & $ 53 $ & $ \frac{4909003176336757275553467075425}{111796641735328007376797249088520192} $ \\ \hline
$ 12 $ & $ \frac{1132693}{1107296256} $ & $ 33 $ & $ \frac{4374403039126240385}{38959523483674573012992} $ & $ 54 $ & $ \frac{19627584575160129028816787257753}{464386050285208646026696265444622336} $ \\ \hline
$ 13 $ & $ \frac{4490513}{5234491392} $ & $ 34 $ & $ \frac{17705045340400677607}{165577974805616935305216} $ & $ 55 $ & $ \frac{78505240133588264624896189049521}{1927640208731054757091946762222960640} $ \\ \hline
$ 14 $ & $ \frac{17337685}{24427626496} $ & $ 35 $ & $ \frac{71484177460946258777}{702452014326859725537280} $ & $ 56 $ & $ \frac{314123632091141305526902518303973}{7996137162143634547936964346998947840} $ \\ \hline
$ 15 $ & $ \frac{65867621}{112742891520} $ & $ 36 $ & $ \frac{287529593953850293471}{2975090884207876484628480} $ & $ 57 $ & $ \frac{1257288843192384664389299749835521}{33147623144886339580538688565741092864} $ \\ \hline
$ 16 $ & $ \frac{249437227}{515396075520} $ & $ 37 $ & $ \frac{1151710503160001680385}{12580384310364734849286144} $ & $ 58 $ & $ \frac{1677695623930304081656255827713551}{45775289104843040373124855638404366336} $ \\ \hline
$ 17 $ & $ \frac{955110593}{2336462209024} $ & $ 38 $ & $ \frac{4596336312298962012663}{53117178199317769363652608} $ & $ 59 $ & $ \frac{20147683002193594117896886735926057}{568577275196997764634603470034917392384} $ \\ \hline
$ 18 $ & $ \frac{3740591431}{10514079940608} $ & $ 39 $ & $ \frac{18298456303802689186745}{223953508083610054614319104} $ & $ 60 $ & $ \frac{26879884904186172110556704720248631}{784244517513100365013246165565403299840} $ \\ \hline
$ 19 $ & $ \frac{15039656569}{47004122087424} $ & $ 40 $ & $ \frac{72784234597284215364691}{942962139299410756270817280} $ & $ 61 $ & $ \frac{322682332818808295011085893297500673}{9729948929145584189655867681252122296320} $ \\ \hline
$ 20 $ & $ \frac{61727254227}{208907209277440} $ & $ 41 $ & $ \frac{289698911730110389042529}{3965276688335983693036257280} $ & $ 62 $ & $ \frac{1290942327848947576849492154270349133}{40217122240468414650577586415842105491456} $ \\ \hline
$ 21 $ & $ \frac{255640084561}{923589767331840} $ & $ 42 $ & $ \frac{1155125274097244765650075}{16654162091011131510752280576} $ & $ 63 $ & $ \frac{5163832046366445947035366917883833877}{166142865649148204785992652078560829243392} $ \\ \hline
$ 22 $ & $ \frac{50273131919}{193514046488576} $ & $ 43 $ & $ \frac{4616317010648384103125561}{69866240967168649264619323392} $ & $ 64 $ & $ \frac{983546099095446058993477411998292607}{32667107224410092492483962313449748299776} $ \\ \hline
\end{longtable}
\endgroup

\subsection{Comparison with~\cite{Kong1}.}
We already compared our result to the previously known best bound used power series expansion from~\cite{ButeraFederbushPernici}.
However, another method of computing lower and upper bounds for $f_2(p)$ based on the empirically observed inequality~\cite[Eq. 16]{Kong1} for strips was proposed in~\cite{Kong1}.
In this paper bounds for $p = \frac{1}{20}, \ldots, \frac{20}{20}$ were computed (see~\cite[Table II]{Kong1}).
We compare our results with this computation in Table~\ref{table:Kong}.
The table shows that for $p$ close to one Kong's results may be more accurate.
On the other hand, our bound is much more precise for $p \leqslant \frac{17}{20}$.

\begingroup
\renewcommand\arraystretch{1.2}
\begin{table}[h!]
\begin{tabular}{|c|l|l|}
\hline
$p$ & \cite{Kong1} & Our estimate \\
\hline
$10/20$ & $0.633195588930[4-5]$ & $\mathbf{0.633195588930}5251415416[5-6]$ \\
\hline
$11/20$ & $0.650499726669[5-8]$ & $\mathbf{0.650499726669}5759205[7-8]$ \\
\hline
$12/20$ & $0.66044120984[2-5]$ & $\mathbf{0.66044120984}322136[2-4]$ \\
\hline
$13/20$ & $0.6625636470[2-4]$ & $\mathbf{0.6625636470}3101[3-4]$ \\
\hline
$14/20$ & $0.65620036[0-1]$ & $\mathbf{0.65620036}1027[4-5]$ \\
\hline
$15/20$ & $0.64039026[3-5]$ & $\mathbf{0.64039026}42[8-9]$ \\
\hline
$16/20$ & $0.6137181[3-4]$ & $\mathbf{0.6137181}37[2-7]$\\
\hline
$17/20$ & $0.573983[2-3]$ & $\mathbf{0.573983}[2-3]$ \\
\hline
$18/20$ & $0.51739[1-2]$ & $\mathbf{0.51739}[1-3]$ \\
\hline
$19/20$ & $0.435[8-9]$ & $\mathbf{0.435}[8-9]$ \\
\hline
\end{tabular}
\caption{Comparison with~\cite{Kong1}. Digits in square brackets mean the corresponding digit in lower and upper bounds.}\label{table:Kong}
\end{table}
\endgroup

\bibliographystyle{ieeetr}
\bibliography{bibdata}

\end{document}